\documentclass{llncs}

\usepackage[english]{babel}
\usepackage[utf8]{inputenc}
\usepackage{amssymb}
\usepackage{amsmath}
\usepackage{color}
\usepackage{graphicx}
\newcommand{\NP}{{\sf NP}}

\newcommand{\W}{{\sf W}}

\begin{document}
\title{
Obtaining Planarity by Contracting Few Edges\thanks{This paper is supported by EPSRC (EP/G043434/1) and Royal Society (JP100692), and by the Research Council of Norway (197548/F20).}}

\author{ 
Petr A. Golovach \inst{1} 
\and  
Pim van~'t Hof \inst{2}
\and
Dani\"el Paulusma \inst{1}
}

\institute{School of Engineering and Computing Sciences, Durham University, UK\\
\texttt{\{petr.golovach,daniel.paulusma\}@durham.ac.uk}
\and
Department of Informatics,
University of Bergen, Norway\\
\texttt{\{pim.vanthof\}@ii.uib.no}}

\maketitle

\begin{abstract}
The {\sc Planar Contraction} problem is to test whether a given graph 
can be made planar by using at most $k$ edge contractions. This problem is known to be \NP-complete. We show that it is fixed-parameter tractable when parameterized by $k$. 
\end{abstract}

\section{Introduction}

Numerous problems in algorithmic graph theory, with a large variety of applications in different fields, can be formulated as graph modification problems. A graph modification problem takes as input an $n$-vertex graph $G$ and an integer $k$, and the question is whether $G$ can be modified into a graph that belongs to a prescribed graph class, using at most $k$ operations of a certain specified type. Some of the most common graph operations that are used in this setting are vertex deletions, edge deletions and edge additions, leading to famous problems such as {\sc Feedback Vertex Set}, {\sc Odd Cycle Transversal}, {\sc Minimum Fill-In} and {\sc Cluster Editing}, to name but a few. More recently, the study of graph modification problems allowing only {\em edge contractions} has been initiated, yielding several results that we will survey below. The contraction of an edge removes both end-vertices of an edge and replaces them by a new vertex, which is made adjacent to precisely those vertices that were adjacent to at least one of the two end-vertices. Choosing edge contraction as the only permitted operation leads to the following decision problem, for each graph class ${\cal H}$.

\medskip
\noindent
{\sc ${\cal H}$-Contraction}\\
{\it Instance:} \hspace*{-.01cm} A graph $G$ and an integer $k$.\\
{\it Question:} Does there exist a graph $H\in {\cal H}$ such that $G$ can be contracted to $H$,\\
\hspace*{1.51cm} using at most $k$ edge contractions?

\medskip
Heggernes et al.~\cite{HHLLP11} presented a $2^{k+o(k)}+n^{O(1)}$ time algorithm for ${\cal H}$-{\sc Con\-traction} when ${\cal H}$ is the class of paths. Moreover, they showed that in this case the problem has a linear vertex kernel. When ${\cal H}$ is the class of trees, they showed that the problem can be solved in $4.98^k n^{O(1)}$ time, and that a polynomial kernel does not exist unless 
\NP\ $\subseteq$ co\NP/poly.
When the input graph is a chordal graph with $n$ vertices and $m$ edges, then ${\cal H}$-{\sc Contraction} can be solved in $O(n+m)$ time when ${\cal H}$ is the class of trees and in $O(nm)$ time when ${\cal H}$ is the class of paths~\cite{HHLeP11}. Heggernes et al.~\cite{HHLP11} proved that ${\cal H}$-{\sc Contraction} is fixed-parameter tractable when ${\cal H}$ is the class of bipartite graphs and $k$ is the parameter. This also follows from a more general result from a recent paper by Marx, O'Sullivan, and Razgon~\cite{MOR11} on generalized bipartization. Golovach et al.~\cite{GKPT11} considered the class of graphs of minimum degree at least $d$ for some integer $d$. They showed that in this case ${\cal H}$-{\sc Contraction} is fixed-parameter tractable when both $d$ and $k$ are parameters, \W[1]-hard when $k$ is the parameter, and para-\NP-complete when $d$ is the parameter.

The combination of planar graphs and edge contractions has been studied before in a closely related setting. Kami\'{n}ski, Paulusma and Thilikos~\cite{KPT10} showed that for every fixed graph $H$, there exists a polynomial-time algorithm for deciding whether a given planar graph can be contracted to $H$. Very recently, this result was improved by Kami\'{n}ski and Thilikos~\cite{KT12}. They showed that, given a graph $H$ and a planar graph $G$, the problem of deciding whether $G$ can be contracted to $H$ is fixed-parameter tractable when parameterized by $|V(H)|$. 

\medskip
\noindent
{\bf Our Contribution.}  
We study ${\cal H}$-{\sc Contraction} when ${\cal H}$ is the class of planar graphs, and refer to the problem as {\sc Planar Contraction}. 
This problem is known to be \NP-complete due to a more general result on ${\cal H}$-{\sc Contraction} by Asano and Hirata~\cite{AH83}. We show that the {\sc Planar Contraction} problem is fixed-parameter tractable when parameterized by $k$. 
This result complements the following results on two other graph modification problems related to planar graphs. The problem of deciding whether a given graph can be made planar by using at most $k$ vertex deletions was proved to be fixed-parameter tractable independently by Marx and Schlotter~\cite{MS12}, who presented a quadratic-time algorithm for every fixed $k$, and by Kawarabayashi~\cite{K09}, whose algorithm runs in linear time for every $k$. Kawarabayashi and Reed~\cite{KR07} showed that deciding whether a graph can be made planar by using at most $k$ edge deletions can also be done in linear time for every fixed $k$.

Our algorithm for {\sc Planar Contraction} starts by finding a set $S$ of at most $k$ vertices whose deletion transforms $G$ into a planar graph. Such a set can be found by using either the above-mentioned linear-time algorithm by Kawarabayashi~\cite{K09} or the quadratic-time algorithm by Marx and Schlotter~\cite{MS12}. The next step of our algorithm is based on the irrelevant vertex technique developed in the graph minors project of Robertson and Seymour~\cite{RS-GMXIII,RS-GMXXII}. 
We show that if the input graph $G$ has large treewidth, we can find an {\em edge} whose contraction yields an equivalent, but smaller instance. After repeatedly contracting such irrelevant edges, we invoke Courcelle's Theorem~\cite{Co90} to solve the remaining instance in linear time.

We finish this section by making two remarks that show that we cannot apply the techniques that were used to prove fixed-parameter tractability of the vertex deletion and edge deletion variants of {\sc Planar Contraction}. First, a crucial observation in the paper of Kawarabayashi and Reed~\cite{KR07} is that any graph that can be made planar by at most $k$ edge deletions must have bounded genus. This 
property is heavily exploited in the case where the treewidth of the input graph is large. The following example shows that we cannot use this technique in our setting. Take a complete biclique $K_{3,r}$ with partition classes $A$ and $B$, where $|A|=3$ and $|B|=r$ for some integer $r\geq 3$. Now make the vertices in $A$ pairwise adjacent and call the resulting graph $G_r$. 
Then $G_r$ can be made modified into a planar graph by contracting one of the edges in $A$. However, the genus of $G_r$ is at least the genus of $K_{3,r}$, which is equal to $\frac{r-2}{2}$~\cite{Bo78}.

Second, the problem of deciding whether a graph can be made planar by at most $k$ vertex deletions for some fixed integer $k$, i.e., $k$ is not part of the input, is called the $k$-{\sc Apex} problem. As observed by Kawarabayashi~\cite{K09} and Marx and Schlotter~\cite{MS12}, the class of so-called $k$-apex graphs (graphs that can be made planar by at most $k$ vertex deletions) is closed under taking minors. This means that the $k$-{\sc Apex} problem can be solved in cubic time for any fixed integer $k$ due to deep results by Robertson and Seymour~\cite{RS04}. However, we cannot apply Robertson and Seymour's result on {\sc Planar Contraction}, because the class of graphs that can be made planar by at most $k$ edge contractions is not closed under taking minors, as the following example shows. Take the complete graph $K_5$ on 5 vertices. For each edge $e=uv$, add a path $P_e$ from $u$ to $v$ consisting of $p$ new vertices for some integer $p\geq k$. 
Call the resulting graph $G^*_p$. Then $G^*_p$ can be made planar by contracting an arbitrary edge of the original $K_5$. However, if we remove all edges of this $K_5$, we obtain a minor of $G^*_p$ that is a subdivision of the graph $K_5$. In order to make this minor planar, we must contract all edges of a path $P_e$, so we need at least $p+1>k$ edge contractions.

\section{Preliminaries}\label{s-pre}

Throughout the paper we consider undirected finite graphs that have no loops and no multiple edges. Whenever we consider a graph problem, we use $n$ to denote the number of vertices of the input graph. 
We refer to the text book of Diestel~\cite{Diestel}  for undefined graph terminology and to the monographs of Downey and Fellows~\cite{DF99} 
for more information on parameterized complexity.

Let $G=(V,E)$ be a graph and let $S$ be a subset of $V$. We write $G[S]$ to denote the subgraph of $G$ {\it induced} by $S$, i.e.,
the subgraph of $G$ with vertex set $S$ and edge set $\{uv\; |\; u,v\in S\; \mbox{with}\; uv\in E\}$. We write $G-S=G[V\setminus S]$, and for any subgraph $H$ of $G$, we write $G-H$ to denote $G-V(H)$. We say that two disjoint subsets $U\subseteq V$ and $W\subseteq V$ are {\it adjacent} if there exist two vertices $u\in U$ and $w\in W$ such that $uw\in E$. Let $H$ be a graph that is not necessarily vertex-disjoint from $G$. Then $G\cup H$ denotes the graph with vertex set $V(G)\cup V(H)$ and edge set $E(G)\cup E(H)$, and $G\cap H$ denotes the graph with vertex set $V(G)\cap V(H)$ and edge set $E(G)\cap E(H)$.

The \emph{contraction} of edge $uv$ in $G$ removes $u$ and $v$ from $G$, and replaces them by a new vertex made adjacent to precisely those vertices that were adjacent to $u$ or $v$ in $G$. A graph $H$ is a {\em contraction} of $G$ if $H$ can be obtained from $G$ by a sequence of edge contractions. 
Alternatively, 
we can define a contraction of $G$ as follows. An {\it $H$-witness structure} $\cal W$ is a partition of $V(G)$ into $|V(H)|$ nonempty sets $W(x)$, one for each $x \in V(H)$, called {\it $H$-witness sets}, such that each $W(x)$ induces a connected subgraph of $G$, and for all $x,y\in V(H)$ with $x\neq y$, the sets $W(x)$ and $W(y)$ are adjacent in $G$ if and only if $x$ and $y$ are adjacent in $H$. Clearly, $H$ is a contraction of $G$ if and only if $G$ has an $H$-witness structure; $H$ can be obtained by contracting each witness set into a single vertex. A {\em witness edge} is an edge of $G$ whose end-vertices belong to two different witness sets. 

Let ${\cal W}$ be an $H$-witness structure of $G$. For our purposes, we sometimes have to contract edges in $G$ such that the resulting graph does {\it not} contain $H$ as contraction. In order to do this it is necessary to {\it destroy} ${\cal W}$, i.e., to contract at least one witness edge in ${\cal W}$. After all, every edge that is not a witness edge has both its end-vertices in the same witness set of ${\cal W}$, which means that contracting such an edge yields an $H$-witness structure of a contraction of $G$. Hence, contracting all such non-witness edges transforms $G$ into $H$ itself. Note that if we destroy ${\cal W}$ by contracting a witness edge $e$ in ${\cal W}$, the obtained graph still has $H$ as a contraction if $e$ was a non-witness edge in some other $H$-witness structure of $G$. 
Hence, in order to obtain a graph that does not have $H$ as a contraction, it is necessary and sufficient to destroy {\em all} $H$-witness structures of $G$.

A \emph{planar graph} $G$ is a graph that can be embedded in the plane, i.e., that can be drawn in the plane so that its edges intersect only at their end-vertices. A graph that is actually drawn in such a way is called a \emph{plane graph}, or an {\em embedding} of the corresponding planar graph. A plane graph $G$ partitions the rest of the plane into a number of connected regions, called the \emph{faces} of $G$. Each plane graph has exactly one unbounded face, called the \emph{outer} face; all other faces are called \emph{inner faces}. Let $C$ be a cycle in a plane graph $G$. Then $C$ divides the plane in exactly two regions: the {\em outer} region of $C$, containing the outer face of $G$, and the \emph{inner} region of $C$. We say that a vertex $u$ of $G$ lies \emph{inside} $C$ if $u$ is in the inner region of $C$. Similarly, $u$ lies \emph{outside} $C$ if $u$ is in the outer region of $C$. The {\it interior} of $C$ with respect to $G$, denoted $\mbox{interior}_G(C)$, is the set of all vertices of $G$ that lie inside $C$. We also call these vertices {\it interior} vertices of $C$. We say that $C$ {\it separates} the vertices that lie inside $C$ from the vertices that lie outside $C$.
A sequence of mutually vertex-disjoint cycles $C_1, \ldots, C_{q}$ in a plane graph 
is called \emph{nested} if there exist disks $\Delta_1, \ldots, \Delta_{q}$ such that $C_i$ is the boundary of $\Delta_i$ for $i = 1, \ldots, q$, and $\Delta_{i+1} \subset \Delta_{i}$ for $i = 1, \ldots, q-1$. We also refer to such a sequence of nested cycles as {\it layers}. We say that a vertex $u$ lies {\it between} two nested cycles $C_i$ and $C_j$ with $i<j$ if $u$ lies in the inner 
region of $C_i$ and in the outer region of $C_j$.

A graph $G$ contains a graph $H$ as a {\em minor} if $G$ can be modified to $H$ by a sequence of edge contractions, edge deletions and vertex deletions. Note that a graph $G$ contains a graph $H$ as a minor if and only if $G$ contains a subgraph that contains $H$ as a contraction.
The \emph{subdivision} of an edge $e=uv$ in a graph $G$ removes $e$ from $G$ and replaces it by a new vertex $w$ that is made adjacent to $u$ and $v$. A {\it subdivision} of a graph $G$ is a graph obtained from $G$ after performing a sequence of edge subdivisions. In Figure~\ref{f-wall}, three examples of an \emph{elementary wall} are given. The unique 
cycle that forms the boundary of the outer face is called the {\em perimeter} of the wall. A \emph{wall} $W$ of height $h$ is a subdivision of an elementary wall of height $h$ and is well-known to have a unique planar embedding. 
We also call the facial cycle of $W$ corresponding to the perimeter of the original elementary wall the {\em perimeter} of $W$, and we denote this cycle by $P(W)$.

\begin{figure}
\begin{center}
\includegraphics[scale=0.6]{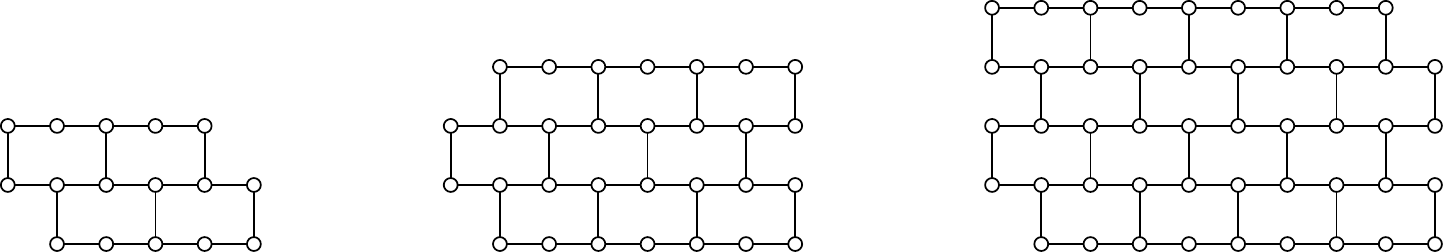} 
\caption{Elementary walls of height 2, 3, and 4 with perimeters of length 14, 22,  and 30, respectively.}\label{f-wall}
\end{center}
\end{figure}

The $r \times r$ {\it grid} has all pairs $(i, j)$ for $i,j = 0,1,\ldots, r-1$ as the vertex set, and two vertices $(i,j)$ and $(i^\prime,j^\prime)$ are joined by an edge if and only if $|i-i^\prime| + |j-j^\prime| = 1$. The \emph{side length} of an $r \times r$ grid is $r$.
A well-known result of Robertson and Seymour~\cite{RS94} states that,
for every integer $r$, any planar graph with no $r\times r$ grid minor has treewidth at most $6r-5$. Although it is not known whether a largest grid minor of a planar graph can be computed in polynomial time, there exist several constant-factor approximation algorithms. 
In our algorithm we will use one by Gu and Tamaki~\cite{GuT11}. 
For any graph $G$, let gm$(G)$ be the largest integer $r$ such that $G$ has an $r \times r$ grid as a minor. 
Gu and Tamaki~\cite{GuT11} showed that for every constant $\epsilon > 0$, there exists a constant $c_{\epsilon}>3$ such that an $r\times r$ grid minor in a planar graph $G$ can be constructed in time 
$O(n^{1 + \epsilon})$, where $r\geq \mbox{gm}(G)/c_\epsilon$.
Because we can obtain a wall of height $\lfloor r/2\rfloor$ as a subgraph from an $r\times r$ grid minor by deleting edges and vertices, their result implies the following theorem.

\begin{theorem}[\cite{GuT11}]\label{t-approx} 
Let $G$ be a planar graph, and let $h^*$ be the height of a largest wall that appears as a subgraph in $G$. For every constant $\epsilon > 0$, there exists a constant $c_{\epsilon}>3$ such that a wall in $G$ with height at least $h^*/c_\epsilon$ can be constructed in time $O(n^{1 + \epsilon})$.
\end{theorem}

In parameterized complexity theory, we consider the problem input as a pair $(I,k)$, where $I$ is the main part and $k$ the parameter. A problem is \emph{fixed-parameter tractable} if an instance $(I,k)$ can be solved in time $f(k)|I|^c$, where $f$ denotes a computable function that only depends on $k$, and where $c$ is a constant independent of $k$.

\section{Fixed-Parameter Tractability of {\sc Planar Contraction}}\label{s-main}

For our algorithm 
we need the 
aforementioned 
result of
Kawarabayashi~\cite{K09}.

\begin{theorem}[\cite{K09}]\label{t-vertex}
For every fixed integer~$k$, it is possible to find in $O(n)$ time a set $S$ of at most $k$ vertices in an $n$-vertex graph $G$ such that $G-S$ is planar, or conclude that such a set $S$ does not exist.
\end{theorem}

We also need the three following lemmas.

\begin{lemma}\label{l-vertexcontract}
If a graph $G=(V,E)$ can be contracted to a planar graph by using at most $k$ edge contractions, then
there exists a set $S\subseteq V$ with $|S|\leq k$ such that $G-S$ is planar.\footnote{As an aside, we point out that the reverse of this statement is not true. For instance, take a $K_5$ and subdivide each of its edges $p\geq 3$ times. The resulting graph can be made planar by one vertex deletion, but at least $p-1$ edge contractions are required.} 
\end{lemma}

\begin{proof}
Let $G=(V,E)$ be a graph on $n$ vertices that can be contracted to a planar graph $H$ by using $k'\leq k$ edge contractions.
Because each edge contraction reduces the number of vertices by exactly one, $H$ has $\ell=n-k'$ vertices $x_1,\ldots,x_\ell$.
Let ${\cal W}$ be an $H$-witness structure of $G$. We write $|W(x_i)|=w_i$ for $i=1,\ldots,\ell$. Note that $w_1+\ldots +w_\ell=n$.

For each witness set $W(x_i)$, we arbitrarily remove a set $S_i$ of $w_i-1$ vertices from $W(x_i)$.
We let $S=S_1\cup \cdots \cup S_\ell$. Because each witness set $W(x_i)$ has size 1 after removing the vertices of $S_i$, 
the graph $G-S$ is a spanning subgraph of $H$. Because $H$ is planar and the class of planar graphs is closed under edge deletion, this means that $G-S$ is planar. Moreover, $|S|=w_1-1+\ldots +w_\ell-1=n-\ell=k'\leq k$. This completes the proof of Lemma~\ref{l-vertexcontract}.\qed 
\end{proof}

When $k$ is fixed, we write {\sc $k$-Planar Contraction} instead of {\sc Planar Contraction}. A seminal result of Courcelle~\cite{Co90} states that on any class of graphs of bounded treewidth, every problem expressible in monadic second-order logic can be solved in time linear in the number of vertices of the graph.

\begin{lemma}\label{l-msol}
For every fixed integer $k$, the {\sc $k$-Planar Contraction} problem can be expressed in monadic second-order logic.
\end{lemma}

\begin{proof}
By Kuratowski's Theorem (cf.~\cite{Diestel}), 
a graph is planar if and only if it does not contain $K_5$ or $K_{3,3}$ as a minor, i.e., if it cannot be modified into $K_5$ or $K_{3,3}$ by a sequence of edge contractions, edge deletions and vertex deletions. 
Recall that a graph $G$ contains a graph $H$ as a minor if and only if $G$ contains a subgraph that contains $H$ as a contraction, i.e., that has an $H$-witness structure.  

Let $G$ be a graph and let $k\geq 0$ be an integer.
We claim that $G$ can be contracted to a planar graph by at most $k$ edge contractions if and only if $G$ contains a set of edges 
$\{e_1,\ldots,e_{k'}\}$ for some $k'\leq k$ such that for all $K_5$-witness structures and all $K_{3,3}$-witness structures of all subgraphs of $G$, at least one of the edges in $\{e_1,\ldots,e_{k'}\}$ is a witness edge. 

First suppose that $G$ can be contracted to a planar graph $H$ by at most $k$ edge contractions. 
Let $E'\subseteq E(G)$ be a set of at most $k$ edges such that contracting all edges in $E'$ transforms $G$ into $H$.
Suppose, for contradiction, that there is a subgraph $F$ of $G$ that has a $K_5$-witness structure or a $K_{3,3}$-witness structure ${\cal W}$ for which none of the edges in $E'$ is a witness edge. Then the graph $F'$ that is obtained from $F$ by contracting all edges in $E'\cap E(F)$ also has a $K_5$-witness structure or a $K_{3,3}$-witness structure. Since $F'$ is a subgraph of $H$, this means that $H$ contains $K_5$ or $K_{3,3}$ as a minor, contradicting the assumption that $H$ is planar.

For the reverse direction, suppose that $G$ contains a set of edges $\{e_1,\ldots,e_{k'}\}$ for some $k'\leq k$ such that for all $K_5$-witness structures and all $K_{3,3}$-witness structures of all subgraphs of $G$, at least one of the edges in $\{e_1,\ldots,e_{k'}\}$ is a witness edge. Let $H$ be the graph obtained from $G$ by contracting each of the edges $e_1,\ldots,e_{k'}$, 
and let ${\cal W}$ be an $H$-witness structure of $G$.
We claim that $H$ is planar. Suppose, for contradiction, that $H$ is not planar. Then $H$ contains a subgraph $F'$ that has
a $K_5$-witness structure or a $K_{3,3}$-witness structure ${\cal W}'$. We consider the subgraph $F$ of $G$ induced by 
the union of 
the $H$-witness sets $W(x)$ of ${\cal W}$ with $x\in V(F')$. For each witness set $W'(a)$ of ${\cal W}'$ we define 
the set $W^*(a)$ to be the union of 
the $H$-witness sets $W(x)$ of ${\cal W}$ with $x\in W'(a)$.
Then the sets $W^*(a)$ form a $K_5$-witness structure or a $K_{3,3}$-witness structure of $F$ in which none of the edges $e_1,\ldots,e_{k'}$ is a witness edge, contradicting our assumption on the set 
$\{e_1,\ldots,e_{k'}\}$.

We observe that for any fixed graph $H$, the property ``having $H$ as a subgraph'' can be expressed in monadic second-order logic. Also ``having an $H$-witness structure'' can be expressed in monadic second-order logic, as we can express the properties ``being connected'' and ``being adjacent'' in monadic second-order logic.
This completes the proof of Lemma~\ref{l-msol}.\qed 
\end{proof}

It is possible that the following lemma or a close variant of it is already known in the literature. However, because we could not find a reference, we give its proof here.

\begin{lemma}\label{l-combined}
Let $B$ be a planar graph that has an embedding with two nested cycles $C_1$ and $C_2$, such that $C_1$ is the boundary of its outer face and $C_2$ is the boundary of an inner face, and such that there are at least two vertex-disjoint paths that join vertices of $C_1$ and $C_2$. Let $I$ be a graph with $B\cap I=C_2$ 
such that $R=B\cup I$ is planar. Then $R$ has an embedding
such that $C_1$ is the boundary of the outer face.
\end{lemma}

\begin{proof}
Let $\cal B$ be a plane graph that is an embedding of $B$ such that $C_1$ is the boundary of the outer face, and $C_2$ is the boundary of an inner face.
Let $\cal R$ be an embedding of $R$ such that $C_2$ lies inside $C_1$ and, subject to this condition, the set of vertices and edges of $R$ that lie in the outer region of $C_1$ has minimum size. To simplify our arguments, we assume that each edge of $\cal R$ that lies outside $C_1$ has at least one endvertex not on $C_1$, i.e., that also lies outside $C_1$; otherwise we can subdivide each edge outside $C_1$ that has both its end-vertices in $C_1$ without loss of generality.

Let $X$ be the set of vertices of ${\cal R}$ that are in the outer 
region of $C_1$. If $X=\emptyset$, then we have the desired embedding. 
Suppose that $X\neq\emptyset$. Let $H$ be a connected component of ${\cal R}[X]$. Let $u_1,\ldots,u_k$ be the vertices of $C_1$ that are adjacent to $H$,  
and assume that the vertices $u_1,\ldots,u_k$ appear on $C_1$ in this order (see also Figure~\ref{fig:R}).
Note that $k\geq 2$, as otherwise $H$ could be embedded inside $C_1$, contradicting the minimality of $X$. 

We claim that $H$ contains no vertex of $I$. In order to obtain a contradiction, suppose that $I\cap H$ contains a vertex $s$. 
Because $H$ is connected and every $u_i$ has a neighbor in $H$ by definition, the subgraph of ${\cal R}$ induced by $V(H)\cup \{u_1,\ldots,u_k\}$ is connected. Hence it contains a path from $s$ to a vertex in $B-C_2$, namely $u_1$ (or any other vertex $u_j$). This path contains no vertex of $C_2$. Consequently, it must contain an edge between a vertex of $I-C_2$ and a vertex of $B-C_2$. This is not possible, because $B\cap I=C_2$.

Now let $Y$ be the set of vertices and edges of ${\cal R}$ that lie in the inner 
region of $C_2$. Let ${\cal B}'$ be the plane graph obtained from $\cal B$ after removing all vertices of $X\setminus V(H)$ together with their incident edges and all elements of $Y$. 
If we remove all vertices of $H$ from ${\cal B}'$, we obtain a plane graph ${\cal B}''$. 
Because $H$ contains no vertex of $I$, all its vertices belong to $B-(C_1\cup C_2)$. Then, by the definition of ${\cal B}$, $H$ must be embedded in the inner region of $C_1$ in 
${\cal B}$, and hence also in ${\cal B}'$. 
Moreover, because there are at least two vertex-disjoint paths that join $C_1$ and $C_2$,
there exists a face in ${\cal B}''$ whose boundary contains a cycle 
$C\notin \{C_1,C_2\}$ 
such that all vertices and edges of $H$ lie in the inner region of $C$ in ${\cal B}'$. 
By the definition of $H$, the only vertices of $B$ that are adjacent to vertices of $H$ are the vertices $u_1,\ldots,u_k$.
As a result, cycle $C$ contains the vertices $u_1,\ldots,u_k$, and they appear in the same order on $C$ as they do on $C_1$.
Moreover, $C$ does not contain any vertex of $H$, as $C$ is a cycle in ${\cal B}''$, and ${\cal B}''$ does not contain any vertex of $H$ by definition. Because $C$ also does not contain any vertex from $(X\setminus V(H))\cup Y$, we find that all vertices of $C$ not on $C_1$ or $C_2$ lie between $C_1$ and $C_2$ in ${\cal R}$ (see Figure~\ref{fig:R} for an illustration). 

\begin{figure}[ht]
\centering
\includegraphics[scale=.9]{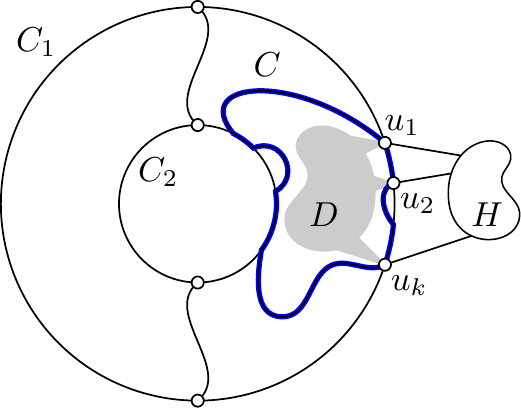}
\caption{An example of a plane graph $\cal R$, in which $H$ is the only 
connected component of $R[X]$. The cycle $C$ is drawn as the bold (blue) cycle.
The shaded area indicates the empty region $D$, in which the connected component $H$ can be embedded.
\label{fig:R}}
\end{figure}

Let $f$ denote the inner region of $C$ in $\cal R$. We say that a path in $\cal R$ that joins two vertices of $C$ and that lies in $f$ is a \emph{chord} of $C$. A chord $P$ {\it separates} $u_i$ and $u_j$ if $u_i$ and $u_j$ are not on the boundary of the same inner face of $C\cup P$ in ${\cal R}$.

\medskip
\noindent
{\it Claim 1. The graph $\cal R$ has no chords separating $u_i$ and $u_j$ for all $1\leq i<j\leq k$.}

\medskip
\noindent
To obtain a contradiction, let $x$ and $y$ be the end-vertices of a chord $P$ in ${\cal R}$ that separates $u_i$ and $u_j$ for some $1\leq i<j\leq k$. By definition, $x,y\in V(C)$. Because $P$ separates $u_i$ and $u_j$, at least one of the vertices $x,y$ does not belong to $C_2$. We assume without loss of generality that $x$ is not on $C_2$. Hence, $x$ is a vertex of $B-I$.

We claim that all vertices and edges of $P$ belong $B$. As a matter of fact, we prove something stronger by showing that all vertices of 
$P-\{y\}$ belong to $B-I$, and that $y$ either belongs to $B-I$ or to $C_2$. In order to see this, suppose that not all vertices of $P$ belong to $B-I$. Let $z$ be the first vertex on $P$ that belongs to $I$ when we start traversing $P$ from $x$. Because $B\cap I=C_2$, we find that $z$ is on $C_2$. Recall that in ${\cal R}$ all vertices inside $C$ lie between $C_1$ and $C_2$ as shown in Figure~\ref{fig:R}. Hence $z=y$, which implies that $y$ belongs to $C_2$.

Because all vertices and edges of $P$ belong to $B$, we find that $P$ is embedded inside $C_1$ in ${\cal B}$ by definition. However, then the connected subgraph $H$ cannot be embedded in the inner region of $C_1$ in ${\cal B}$. 
This contradiction completes the proof of Claim~1.

\medskip
\noindent
For $k=2$, Claim 1 immediately implies that the embedding $\cal R$ of $R$ can be modified in such a way that $H$ is placed in $f$, as we can find an empty open disk $D$ inside 
$C$, such that no vertex or edge of ${\cal R}$ lies in $D$ and the boundary of $D$ contains $u_1,\ldots,u_k$ (see Figure~\ref{fig:R} for an illustration). This contradicts the minimality of $X$. 
If $k\geq 4$, then there cannot exist a path in $f$ that 
joins two vertices $u_i,u_j$ such that 
$2\leq |i-j|\leq k-2$ as a result of Claim 1.  
Hence we can, like in the case $k=2$, find an empty open disk $D$ inside 
$f$, and obtain a new embedding of $R$ by placing $H$ inside $D$.

For the case $k=3$, we need the following claim in addition; we can prove this claim by the same arguments as we used in the proof of Claim~1.

\medskip
\noindent
{\it Claim 2. There is no vertex $v$ in $\cal R$ that is in $f$ such that $v$ is joined with $u_1,u_2,u_3$ by disjoint paths in $f$.}

\medskip
\noindent
Claims 1 and 2 together imply that we can find a desired empty region $D$ inside $C$ in the case $k=3$, just as in the cases
$k=2$ and $k\geq 4$. 
This completes the proof of Lemma~\ref{l-combined}.\qed
\end{proof}

We are now ready to present our main theorem, which shows that {\sc Planar Contraction} is fixed-parameter tractable when parameterized by $k$. At some places in our proof of Theorem~\ref{t-main} we allow constant factors (independent of $k$) to be less than optimal in order to make the arguments easier to follow.

\begin{theorem}\label{t-main}
For every fixed integer $k$ and every constant $\epsilon>0$, the {\sc $k$-Planar Contraction} problem can be solved in $O(n^{2+\epsilon})$ time.
\end{theorem}

\begin{proof}
Let $G$ be a graph on $n$ vertices, and let $k$ be some fixed integer. If $G$ has connected components $L_1,\ldots,L_q$ for some $q\geq 2$, then we solve for every possible tuple $(k_1,\ldots,k_q)$ with $\sum_{i=1}^qk_i=k$ the instances $(L_1,k_1),\ldots, (L_q,k_q)$. Hence, we may assume without loss of generality that $G$ is connected.
We apply Theorem~\ref{t-vertex} to decide in $O(n)$ time whether $G$ contains a subset $S$ of at most $k$ vertices such that $G-S$ is planar. If not, then we return {\tt no} due to Lemma~\ref{l-vertexcontract}.
Hence, 
from now on we assume that we have found such a set $S$. 
We write $H=G-S$. 

Choose $\epsilon>0$. We apply Theorem~\ref{t-approx} on the graph $H$ to find in $O(n^{1+\epsilon})$ time a subgraph $W$ of $H$ that is a wall with height $h\geq h^*/c_\epsilon$, where $h^*$ denotes the height of a largest wall in $H$ and $c_\epsilon>3$ is some constant.

Suppose that 
$h\leq  \lceil \sqrt{2k+1}\rceil(12k+10)$. Then $h^*\leq c_\epsilon h \leq c_\epsilon \lceil\sqrt{2k+1}\rceil(12k+10)$, i.e., the height of a largest wall in $H$ is bounded by a constant. Consequently, the treewidth of $H$ is bounded by a 
constant~\cite{RS94}.
Since deleting a vertex from a graph decreases the treewidth by at most~1, the treewidth of $G$ is at most $|S|\leq k$ larger than the treewidth of $H$. 
Because $k$ is fixed, this means that the treewidth of $G$ is bounded by a constant as well. Then 
Lemma~\ref{l-msol} tells us that we may apply  Courcelle's Theorem~\cite{Co90} to check in $O(n)$ time if $G$ can be modified into a planar graph by using at most $k$ edge contractions. 

\begin{figure}
\centering
\includegraphics[scale=0.63]{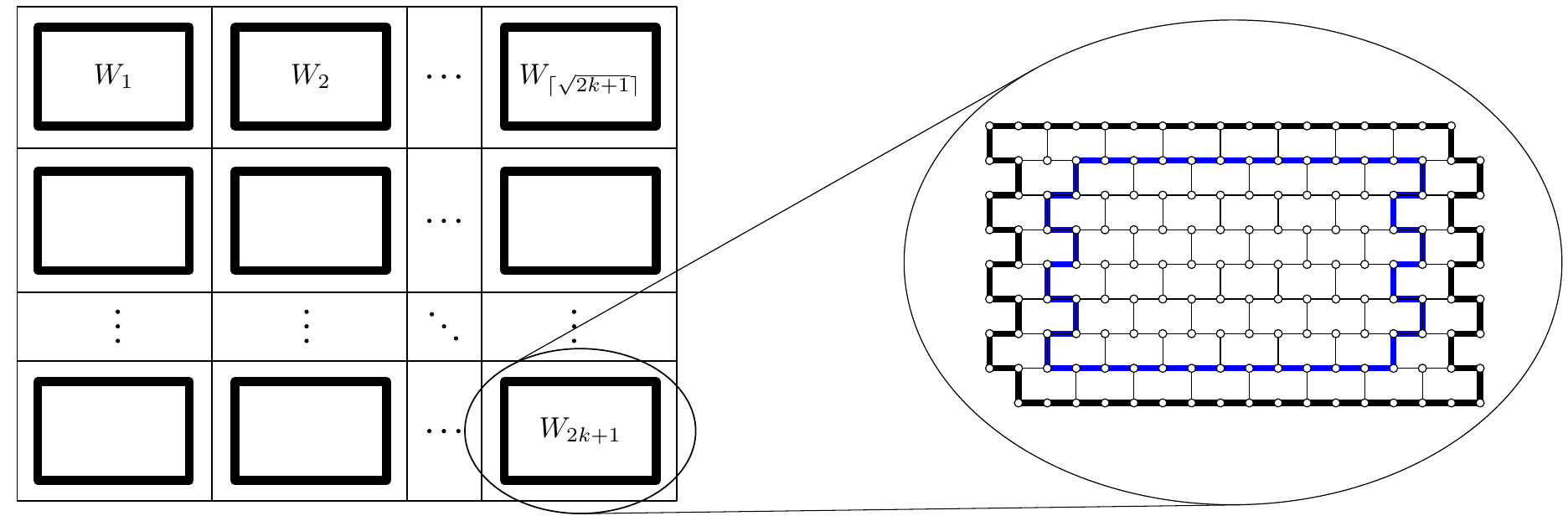}
\caption{
On the left, 
a schematic depiction of the wall $W$ with height 
$h > \lceil \sqrt{2k+1}\rceil(12k+10)$ and the way the subwalls $W_1,\ldots,W_{2k+1}$, each with height 
$12k+8$, are packed within $W$. On the right, a more detailed picture of a subwall $W_i$ of height 
8 in case $k=0$. The bold blue edges indicate the perimeter of the smaller subwall $W_i'$ of height 6.}
\label{f-wallpacking}
\end{figure}

Now suppose that 
$h > \lceil \sqrt{2k+1}\rceil(12k+10)$. We consider some fixed planar embedding of $H$. For convenience, whenever we mention the graph $H$ below, we always refer to this fixed embedding. The wall $W$ is contained in some connected component $\tilde{H}$ of $H$, and we assume without loss of generality that all other connected components of $H$ lie outside $P(W)$. Inside $P(W)$, we choose $2k+1$ mutually vertex-disjoint subwalls $W_1,\ldots, W_{2k+1}$ of height $12k+8$ that are packed inside $W$ in $\lceil \sqrt{2k+1}\rceil$ rows of $\lceil \sqrt{2k+1}\rceil$ subwalls, such that vertices of distinct subwalls are not adjacent; see Figure~\ref{f-wallpacking}.
Inside each $W_i$, we choose a subwall $W_i'$ of height $12k+6$ such that the perimeters of $W_i$ and $W_i'$ are vertex-disjoint; see Figure~\ref{f-wallpacking} for a depiction of $W_i$ and $W_i'$ in case $k=0$.
By definition, the inner region of $P(W_i)$ is the region that contains 
the vertices of $W_i'$, and the inner region of $P(W_i')$ is the region that contains no vertex of $P(W_i)$. 
Note that the interiors of $P(W_i)$ and $P(W_i')$ are defined with respect to (the fixed planar embedding of) the graph $H$.
Hence, these interiors may contain vertices of $H$ that do not belong to $W$, as $W$ is a subgraph of $H$.

We now consider the graph $G$. Recall that $H=G-S$, and that $G$ is not necessarily planar. Hence, whenever we speak about the interior of some cycle below, we always refer to the interior of that cycle with respect to the fixed planar embedding of $H$. For $i=1,\ldots,2k+1$, let $S_i\subseteq S$ be the subset of vertices of $S$ that are adjacent to an interior vertex of $P(W_i')$. Observe that the sets $S_i$ are not necessarily disjoint, since a vertex of $S$ might be adjacent to interior vertices of $P(W_i')$ for several values of $i$. Also note that no vertex of $S$ belongs to $W$, since $W$ is a wall in the graph $H=G-S$. We can construct the sets $S_i$ in $O(n)$ time, because the number of edges of $G$ is $O(n)$. The latter can be seen as follows.
The number of edges in $G$ is equal to the sum of the number of edges of $H$, the number of edges between $H$ and $S$, and the number of edges of $G[S]$. Because $H$ is planar, the number of edges of $H$ is at most $5|V(H)|\leq 5n$. 
Hence, the number of edges of $G$ is at most $5n+kn+\frac{1}{2}k(k-1)=O(n)$ for fixed $k$. 

We say that a set $S_i$ is of {\it type 1} if $S_i$ is non-empty and if every vertex $y\in S_i$ also belongs to some set $S_j$ for $j\neq i$, i.e., every vertex $y\in S_i$ is adjacent to some vertex $z$ that lies inside $P(W'_j)$ for some $j\neq i$; see Figure~\ref{f-witnessstructure} for an illustration. Otherwise we say that $S_i$ is of {\it type 2}. We can check in $O(n)$ time how many sets $S_i$ are of type 1. We claim the following.

\medskip
\noindent
{\it Claim 1. If there are at least $k+1$ sets $S_i$ of type 1, then $(G,k)$ is a no-instance}.

\medskip
\noindent
We prove Claim 1 as follows. Suppose that there exist $\ell\geq k+1$ sets $S_i$ of type~1, say these sets are $S_1,\ldots,S_\ell$. Then for each $i=1,\ldots,\ell$ we can define a $K_5$-witness structure ${\cal X}_i$ of a subgraph of $G$ as follows. We divide the perimeter of $W_i$ into three connected non-empty parts in the way illustrated in Figure~\ref{f-witnessstructure}. 
\begin{figure}
\centering
\includegraphics[scale=0.95]{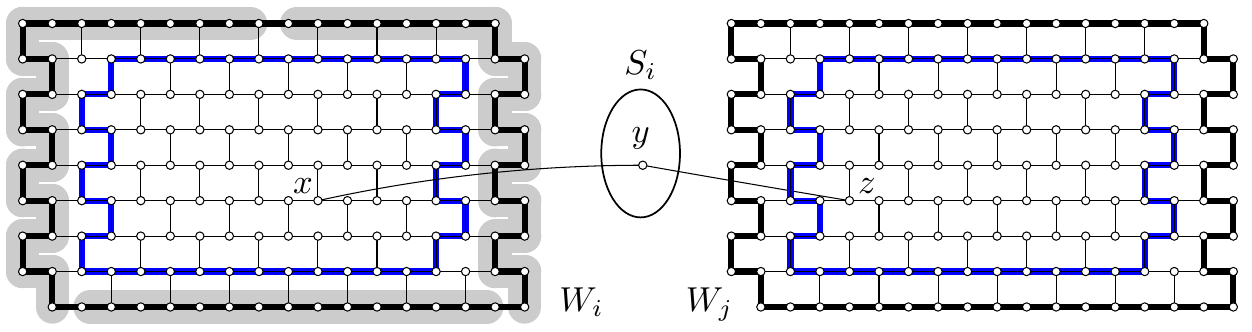}
\caption{
Two subwalls $W_i$ and $W_j$, where the bold blue edges indicate the perimeters $P(W_i')$ and $P(W_j')$ of the smaller subwalls $W_i'$ and $W_j'$. Vertex $y$ is in $S_i$, since it is adjacent to an interior vertex $x$ of $P(W_i')$. 
If, for every $y\in S_i$, there is an edge between $y$ and an interior vertex $z$ of $P(W'_j)$ for some $j\neq i$, then $S_i$ is of type 1. The three shaded areas indicate how the perimeter of $W_i$ is divided into three non-empty parts, each forming a separate witness set of a $K_5$-witness structure ${\cal X}_i$ of a subgraph of $G$.}
\label{f-witnessstructure}
\end{figure}
The vertices of each part will form a separate witness set of ${\cal X}_i$; let us call these witness sets $X_i^1, X_i^2, X_i^3$. 
Let $H'_i$ be the subgraph of $H$ induced by the vertices that lie inside $P(W_i)$ in $H$. The fourth set $X_i^4$ of the witness structure ${\cal X}_i$ consists of all the vertices of the connected component of $H'_i$ that contains $W_i'$.
Let $G'_i$ be the graph obtained from $G$ by deleting all the vertices of $P(W_i)$ and all the vertices that lie inside $P(W_i)$ in $H$, i.e., $G'_i=G-P(W_i)-\mbox{interior}_H(P(W_i))$. 
Let $D$ be the connected component of $G'_i$ that contains the perimeter $P(W)$ of the large wall $W$. It is clear that $D$ contains all vertices that are on or inside $P(W_j)$ for every $j\neq i$. Hence, due to the assumption that $S_i$ is of type 1, all vertices of $S_i$ also belong to $D$. The fifth set $X_i^5$ is defined to be the vertex set of $D$. Let us argue why these five sets form a $K_5$-witness structure of a subgraph of $G$. 

It is clear that each of the sets $X_i^1,X_i^2,X_i^3$ is connected, and that they are pairwise adjacent. The set $X_i^4$ is connected by definition. The choice of the subwall $W_i'$ within $W_i$ ensures that $X_i^4$ is adjacent to each of the sets $X_i^1,X_i^2,X_i^3$. Let us consider the set $X_i^5$. By definition, $X_i^5$ is connected. Since $X_i^5$ contains the perimeter $P(W)$ of the large wall $W$ and $W_i$ lies inside $P(W)$, set $X_i^5$ is adjacent to $X_i^1$, $X_i^2$ and $X_i^3$. Since $S_i$ is of type 1 and hence non-empty, there is a vertex $y\in S_i$ that is adjacent to a vertex $x$ that lies inside $P(W'_i)$ by the definition of $S_i$. We already argued that $X_i^5$ contains all vertices of $S_i$, so $y\in X_i^5$.
Recall that $\tilde{H}$ is the unique connected component of $H$ that contains the wall $W$, and that all connected components of $H$ other than $\tilde{H}$ were assumed to lie outside $P(W)$ in $H$. Because $x$ lies inside $P(W'_i)$, this means that $x$ is in the connected component of $H'_i$ that contains $W_i'$, implying that $x\in X_i^4$. Consequently, the edge between $x$ and $y$ ensures the adjacency between $X_i^4$ and~$X_i^5$.

We now consider the $\ell$ different $K_5$-witness structures ${\cal X}_i$ of subgraphs of $G$ defined in the way described above, one for each $i\in \{1,\ldots,\ell\}$. 
Let us see how such a $K_5$-witness structure ${\cal X}_i$ can be destroyed by using edge contractions only. 
Denote by $E_i$ the set of edges of $G$ incident with the vertices of $X_i^1\cup\cdots \cup X_i^4$ for $i=1,\ldots,\ell$. 
We can only destroy a witness structure ${\cal X}_i$ by edge contractions if we contract the edges 
of at least one path that has its endvertices in different witness sets of ${\cal X}_i$ 
and its inner vertices (in case these exist) not belonging to any witness set of ${\cal X}_i$.
Clearly, such a path always contains an edge of $E_i$. 
Hence, in order to destroy ${\cal X}_i$, we have to contract at least one edge of $E_i$.
Because the sets $E_1,\ldots,E_{\ell}$ are pairwise disjoint
by the construction of the witness structures ${\cal X}_i$,
we must use at least 
$\ell \geq k+1$ edge contractions in order to make $G$ planar. Hence, $G$ is a no-instance. This proves Claim~1.

\medskip
\noindent
Due to Claim 1, we are done if there are at least $k+1$ sets $S_i$ of type 1. Note that every step in our algorithm so far took $O(n^{1+\epsilon})$ time, as desired.
Suppose that we found at most $k$ sets $S_i$ of type 1. Because the total number of sets $S_i$ is $2k+1$, this means that there are at least $k+1$ sets $S_i$ of type 2. 
Let $S_i$ be a set of type 2. If $S_i$ is non-empty, 
then $S_i$ contains a vertex $x$ that is not adjacent to an interior vertex of $P(W'_j)$ for any $j\neq i$, as otherwise $S_i$ would be of type 1. 
Consequently, $S_i$ is the only set of type 2 that contains $x$. Since $|S|\leq k$ and there are at least $k+1$ sets of type 2, at least one of them must be empty. Without loss of generality, we assume from now on that $S_1=\emptyset$.

We will now exploit the property that $S_1=\emptyset$, i.e., that none of the vertices in the interior of $P(W_1')$ is adjacent to any  vertex of $S$. We define a {\it triple layer} as 
the perimeter of a wall with the perimeters of its two largest 
proper subwalls inside, 
such that the three perimeters are mutually vertex-disjoint, and the {\it middle} perimeter is adjacent to the {\it outer} and {\it inner} perimeter; 
see Figure~\ref{f-triplelayer} for an illustration. We define a sequence of nested triple layers in the same way as we defined a sequence of layers in Section~\ref{s-pre}. Because $W_1'$ has height $12k+6$, there exist two adjacent vertices $u$ and $v$ inside $P(W_1')$, such that $u$ and $v$ are separated from the vertices outside $P(W_1')$ by $2k+1$ nested triple layers $L_1,\ldots,L_{2k+1}$, i.e., $u$ and $v$ lie inside the inner perimeter of triple layer $L_{2k+1}$ (see also Figure~\ref{f-triplelayer}). 

\begin{figure}
\centering
\includegraphics[scale=0.6]{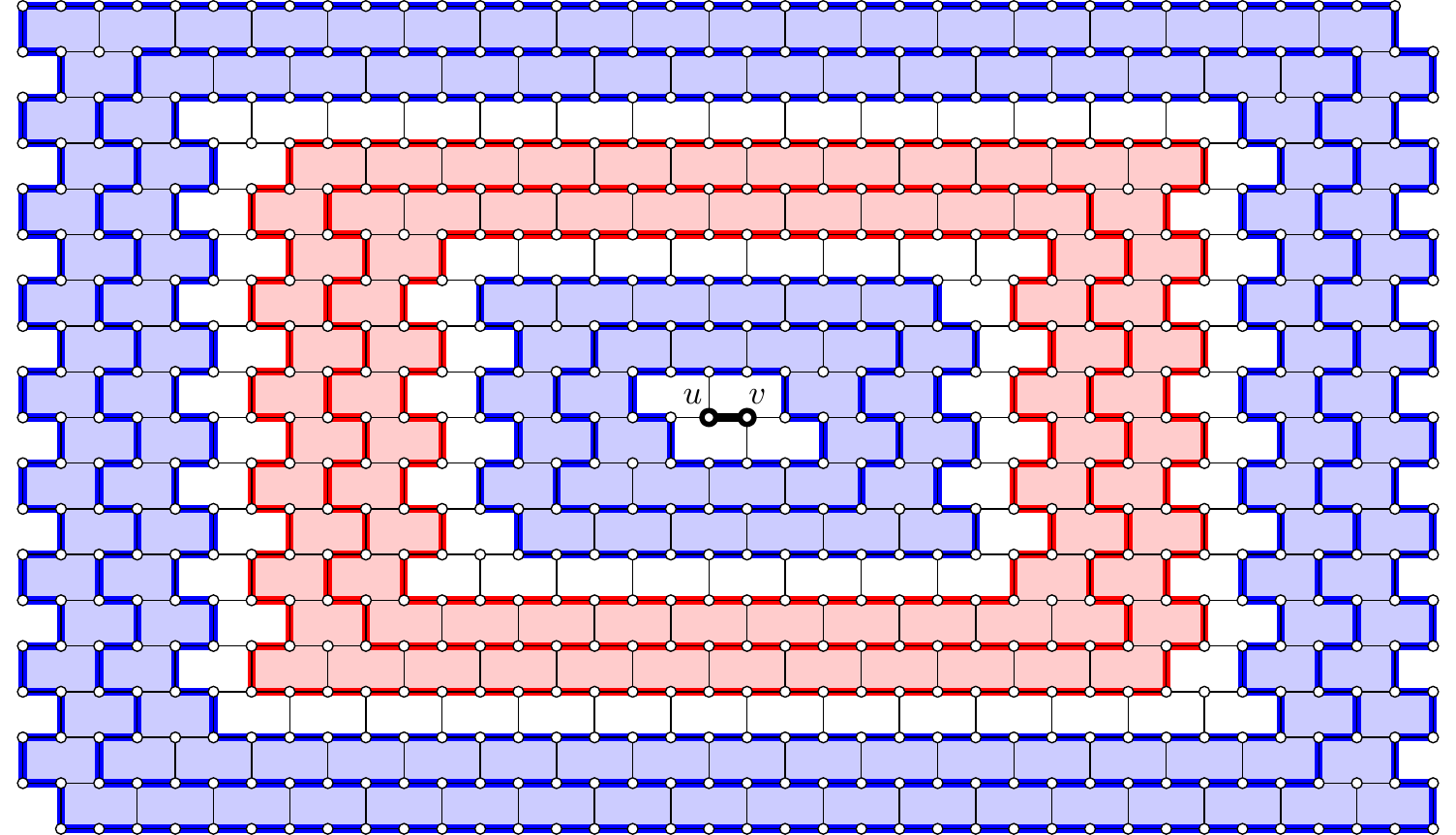}
\caption{The wall $W_1'$ of height $12k+6=18$ in case $k=1$. The black bold edge $uv$ is separated from the vertices that lie outside 
$W_1'$ by $2k+1=3$ nested triple layers. The three shaded areas indicate the sets $Y_1$, $Y_2$, and $Y_3$.}
\label{f-triplelayer}
\end{figure}

Let $G'$ be the graph obtained from $G$ after contracting $uv$. The following claim shows that $uv$ is an ``irrelevant'' edge, i.e., that $uv$ may be contracted without loss of generality.

\medskip
\noindent
{\it Claim 2. $(G,k)$ is a yes-instance if and only if $(G',k)$ is a yes-instance.}

\medskip
\noindent
We prove Claim 2 as follows. First suppose that $(G,k)$ is a yes-instance. This means that $G$ can be modified into a planar graph $F$ by at most $k$ edge contractions. Let $E'\subseteq E(G)$ be a set of at most $k$ edges whose contraction modifies $G$ into $F$. Observe that we can contract the edges in $E'$ in any order to obtain $F$ from $G$. If $uv\in E'$, then we can first contract $uv$ to obtain the graph $G'$, and then contract the other edges in $E'$ to modify $G'$ into the planar graph $F$. If $uv\notin  E'$, then we first contract the edges in $E'$ to modify $G$ into $F$, and then contract the edge $uv$. This leads to a graph $F'$. Since planar graphs are closed under edge contractions, $F'$ is planar. Moreover, $F'$ can also be obtained from $G'$ by contracting the edges in $E'$. We conclude that $(G',k)$ is a yes-instance.

Now suppose that $(G',k)$ is a yes-instance. This means that $G'$ can be modified into a planar graph $F'$ by at most $k$ edge contractions.  Let $E'\subseteq E(G')$ be a set of at most $k$ edges whose contraction modifies $G'$ into 
$F'$. Let $F$ be the graph obtained from $G$ by contracting all the edges of $E'$. 
We will show that $F$ is planar as well.

Recall that $S_1=\emptyset$, and that we defined $2k+1$ triple layers $L_1,\ldots,L_{2k+1}$ inside $P(W_1')$. Let $Q_i$, $Q_i'$, and $Q_i''$ denote the three perimeters in $H$ that form the triple layer $L_i$ for $i=1,\ldots,2k+1$, where $Q_i$ is the outer perimeter, $Q_i'$ the middle 
perimeter, and $Q_i''$ the 
inner perimeter. Let $Y_i$ be the set of all vertices of $H$ that are in $Q_i\cup Q_i'\cup Q_i''$ or that lie in the intersection of the inner 
region of $Q_i$ and the outer region of $Q_i''$, i.e., $Y_i$ is the set of vertices that lie on or ``in between'' the perimeters $Q_i$ and $Q_i''$ in $H$; see Figure~\ref{f-triplelayer} for an illustration.
Because we applied at most $k$ edge contractions in $G'$, there exists a set $Y_i$, for some 
$1\leq i\leq 2k+1$, such that none of its vertices is incident with an edge in $E'$. 
This means that $L_i$ is a triple layer in $F'$ as well.  
We consider a planar embedding of $F'$, in which $Q_i''$ is in the inner 
region of $Q_i'$, and $Q_i'$ is in the inner region of $Q_i$; for convenience, we will denote this planar embedding by $F'$ as well. 

We will now explain how to apply Lemma~\ref{l-combined}. 
We define $C_1$ and $C_2$ to be the perimeters $Q_i'$ and $Q_i''$, respectively. 
We define $B$ as the subgraph of $F'$ 
induced by the vertices that either are in $Q_i'\cup Q_i''$ or lie between $Q_i'$ and $Q_i''$ in $F'$.
Here, we assume that $B$ is connected, as we can always place connected components of $B$ that do not contain vertices
from $Q_i'\cup Q_i''$ outside $Q_i'$.
Because $Q_i'$ and $Q_i''$ are perimeters of subwalls in $H$, and $Q_i''$ is contained inside $Q_i'$, 
there exist at least two vertex-disjoint paths $P_1,P_2$ in $H$ joining $Q_i'$ and $Q_i''$ using vertices of $Y_i$ only. 
Because none of the vertices in $Y_i$ is incident with an edge in $E'$, the two paths $P_1,P_2$ are also vertex-disjoint in $F'$, and consequently in $B$.

We now construct the graph $I$. 
Because $F'$ is a contraction of $G'$, and $G'$ is a contraction of $G$, we find that $F'$ is a contraction of $G$.
Let ${\cal W}$ be an $F'$-witness structure  
of $G$ corresponding to contracting exactly the edges of $E'\cup \{uv\}$ in $G$.   
Then we define $I$ to  be the subgraph of $G$ induced by the union of the vertices of all the witness sets $W(x)$ with
$x$ on or inside $Q_i''$ in $F'$. 
Just as we may assume that $B$ is connected, we may also assume that the subgraph of $F'$ induced by the vertices
that lie on or inside $Q_i''$ is connected. Because witness sets are connected by definition, we then  find that  
$I$ is connected.

Because the edge $uv$ is contracted when $G$ is transformed into $F'$, $u$ and $v$ belong to the same
witness set of ${\cal W}$. 
Let $x^*$ be the vertex of $F'$, such that $u$ and $v$ are in the witness set $W(x^*)$.
Recall that all the vertices of $Q_i''$ and the vertices $u$ and $v$ belong to the wall $W_1'$. 
Since $u$ and $v$ lie inside $Q_i''$ in $H$ and walls have a unique plane embedding, $x^*$ lies 
inside $Q_i''$ in $F'$.
Hence, $u$ and $v$ are vertices of $I$.
Also recall that none of the vertices of $Y_i$, and none of the vertices of $Q_i''$ in particular,  
is incident with an edge of $E'$.
Hence, the vertices of $Q_i''$ correspond to witness sets of ${\cal W}$ that are singletons, i.e., that have cardinality 1. 
This means that we can identify each vertex of $Q_i''$ in $F'$ with the unique vertex of $G$ in the corresponding witness set. Hence,  we obtain that $B\cap I=Q_i''=C_2$.

We now prove that $R=B\cup I$ is planar. For doing this, we first prove that 
$B$ contains no vertex $x$ with $W(x)\cap S\neq\emptyset$, and that $I$ contains no vertex from $S$. 
To see that $B$ contains no vertex $x$ with $W(x)\cap S\neq \emptyset$, assume 
that $x$ is a vertex of $B$ and $s$ is a vertex of $S$ with $s\in W(x)$. 
Recall that no vertex from $Q_i'$ is incident with an edge in $E'$.  Hence, we can identify each vertex in $Q_i'$ in $F'$
with the unique vertex of the corresponding witness set, just as we did earlier with the vertices of $Q_i''$.
Because $s\in W(x)$, this means that $x$ is not in $Q_i'$. 
Because $B$ is connected, we find that
$F'$ contains a path from $x$ to a vertex $y$ in $Q_i'$ that contains no vertex from $Q_i$. 
Note that since $y$ is in $Q_i'$, $y$ is a vertex in $G$ as well.
Because $W(x)$ induces a connected subgraph of $G$ by definition, 
 this path can be transformed into a path in $G$ from $s$ to $y$ that does not contain a vertex from $Q_i$. This is not possible, because $S_1=\emptyset$ implies that every path in $G$ from $s$ to $y$ must go through $Q_i$.
 
We now show that $I$ contains no vertex from $S$. In order to obtain a contradiction, assume that $I$ contains a vertex $s\in S$. Because $I$ is connected, this means that $G$ contains a path from $s$ to a vertex in $Q_i''$ that contains no vertex from $Q_i$ (and also no vertex from $Q_i'$). This is not possible, because $S_1=\emptyset$.

Let $R'$ be the subgraph of $G$ induced by the vertices in the sets
$W(x)$ with $x\in V(B)$ and the vertices of $I$.
Since we proved that $R'$ contains no vertices from $S$, $R'$ is a subgraph of $H$.
Consequently, $R'$ is planar because $H$ is planar. As a result, $R$ is planar because $R$ can be obtained from $R'$ by contracting all edges in every set $W(x)$ with $x\in V(B)$, and planar graphs are closed under edge contractions.
As we have shown that $R=B\cup I$ is planar, we are now ready to apply 
Lemma~\ref{l-combined}. 
This lemma tells us that $R$ has an embedding ${\cal R}$, such that $Q_i'=C_1$ is the boundary of the outer face.
Now consider the embedding that we obtain from the (plane) graph $F'$ by removing all vertices that lie inside $Q_i'$. We combine this embedding with ${\cal R}$ to obtain a plane embedding 
of a graph $F^*$.  
We can obtain $F$ from $F^*$ by contracting all edges in $E'$ that
are incident to a vertex in $I$; recall that $u$ and $v$ are both in $I$ and that $uv$ is not an edge of $E'$.
Because planar graphs are closed under  edge contractions, this means that $F$ is planar.
This completes the proof of Claim~2.

\medskip
\noindent
We can find the irrelevant edge $uv$ mentioned just above Claim~2 in $O(n)$ time. Since all other steps took $O(n^{1+\epsilon})$ time, we used $O(n^{1+\epsilon})$ time so far. After finding the edge $uv$, we contract it and continue with the smaller graph $G'$. Because removing $S$ will make $G'$ planar as well, we can keep $S$ instead of applying Theorem~\ref{t-vertex} again. Hence, we apply Theorem~\ref{t-vertex} only once. Because $G$ has $n$ vertices, and 
every iteration reduces the number of vertices by exactly one,
the total running time of our algorithm is $O(n^{2+\epsilon})$. This completes the proof.
\qed
\end{proof}

\section{Future Work}\label{s-future}

A natural direction for future work is to consider the class ${\cal H}$ that consists of all $H$-minor free graphs for some graph $H$ and
to determine the parameterized complexity of ${\cal H}$-{\sc Contraction} for such graph classes.
Our proof techniques rely on the fact that we must contract to a planar graph, and as such they cannot be used directly for this variant. Hence, we pose this problem as an open problem.

\end{document}